\documentclass[hidelinks, letterpaper, 10 pt, conference]{ieeeconf}
\pdfoutput=1

\usepackage{amsmath}    
\usepackage{amssymb}    
\usepackage{graphicx}    
\usepackage{tikz}        
\usepackage{float}        
\usepackage{soul}         
\usepackage{lmodern}    
\usepackage{caption}
\usepackage{subcaption}
\usepackage{bm}			
\usepackage{cite}

\newtheorem{thm1}{\bf Theorem}
\newtheorem{prop1}{\bf Proposition}
\newtheorem{lem1}{\bf Lemma}
\newtheorem{assmpt1}{\bf Assumption}
\newtheorem{defn1}{\bf Definition}
\newtheorem{rem1}{\bf Remark}

\newtheorem{cor1}{\bf Corollary}

\newenvironment{asm}{\begin{assmpt1}}{\hfill$\Diamond$ \end{assmpt1}}
\newenvironment{rem}{\begin{rem1}}{\hfill$\Diamond$\end{rem1}}
\newenvironment{lem}{\begin{lem1}}{\hfill$\Diamond$\end{lem1}}
\newenvironment{thm}{\begin{thm1}}{\hfill$\Diamond$\end{thm1}}

\newcommand{\bR}{{\mathbb{R}}}

\newcommand{\cC}{{\mathcal{C}}}

\newcommand{\cE}{{\mathcal{E}}}

\newcommand{\cG}{{\mathcal{G}}}

\newcommand{\cL}{{\mathcal{L}}}

\newcommand{\cN}{{\mathcal{N}}}

\newcommand{\cX}{{\mathcal{X}}}

\newcommand{\rd}{{\rm d}}

\DeclareMathOperator*{\argmin}{argmin}

\newcommand*{\argminl}{\argmin\limits}

\newcommand{\rt}[1]{\textcolor{black}{#1}}

\IEEEoverridecommandlockouts                              

\title{\LARGE \bf 
	Distributed Algorithm for Economic Dispatch Problem \\ with Separable Losses
}
\author{Seungjoon Lee and Hyungbo Shim
	\thanks{This work was supported by the National Research Foundation of Korea (NRF) grant funded by the Korea government (Ministry of Science and ICT) (No. NRF-2017R1E1A1A03070342).}
	\thanks{S. Lee and H. Shim are with ASRI, Department of Electrical and Computer Engineering, Seoul National University, Seoul, Korea. {\tt\small seungjoon.lee@cdsl.kr, hshim@snu.ac.kr}}%
}

\begin{document}
	\maketitle
	\thispagestyle{empty}
	\pagestyle{empty}
	
	\begin{abstract}
		Economic dispatch problem for a networked power system has been considered.
		The objective is to minimize the total generation cost while meeting the overall supply-demand balance and generation capacity.
		In particular, a more practical scenario has been studied by considering the power losses.
		A non-convex optimization problem has been formulated where the non-convexity comes from the nonlinear equality constraint representing the supply-demand balance with the power losses.
		It is shown that the optimization problem can be solved using convex relaxation and dual decomposition.
		A simple distributed algorithm is proposed to solve the optimization problem. 
		Specifically, the proposed algorithm does not require any initialization process and hence robust to various changes in operating condition.
		In addition, the behavior of the proposed algorithm is analyzed when the problem is infeasible.
	\end{abstract}

	\section{Introduction}    
	One of the fundamental problems which arise in the operation of the power system is to balance the overall energy demand with generation. 
	In particular, finding the optimal generation has been an important problem due to the socioeconomic impacts of the power system in modern society.
	The problem of finding the optimal generation is termed as the economic dispatch problem (EDP) \cite{Wood2012}. 
	The EDP is often formulated as an optimization problem to decide the power generation of each generator subject to various constraints while minimizing the generation cost. 
	Constraints must include the overall supply and demand balance, while additional constraints such as local generation capacity are often imposed. 
	
	The EDP has been studied extensively for the past several decades. 
	Early works were focused on developing centralized algorithms for the EDP. 
	For example, numerical methods \cite{Gaing2003} and Lagrangian relaxation \cite{Guo1996} are developed to solve the EDP.
	However, the power network is growing with the introduction of the smart grid and distributed energy resources.
	Hence, there have been significant efforts in recent years to develop a distributed algorithm to solve the EDP due to its scalability and resiliency.
	Various discrete-time algorithms have been proposed in the literature to solve the EDP in a distributed manner \cite{Elsayed2015,Yang2013,Kar2014}.
	However, most of these works are not suitable for plug-and-play operation due to requiring an initialization process \cite{Elsayed2015,Yang2013} or decaying step sizes \cite{Kar2014}.
	
	On the other hand, continuous-time algorithms are also studied due to the ease of applying classical results on the stability of nonlinear systems \cite{Ahn2018,Cherukuri2016,Yi2016,Yun2019}.
	In \cite{Ahn2018}, authors considered the EDP with power distribution, but it requires an initialization process and had no capacity constraints.
	Initialization-free algorithms are proposed in \cite{Cherukuri2016,Yi2016} which employed dynamic average consensus and in \cite{Yun2019} using dual decomposition and strong coupling.
	
	Most of the works for the EDP mentioned so far considered an ideal scenario where there are no losses in the power system.
	However, there are various forms of losses occurring in the power network which are significant to the operation of the power system.
	For example, losses of an electrical generator such as copper losses or core losses are up to $10\%$ of generation depending on the operating condition \cite{Dutta2013}.
	In addition, transmission and distribution of the power also cause losses, further lowering the overall efficiency. 
	Therefore, it is important to solve the EDP considering the losses of the system.
	
	Solutions to the EDP with power losses mainly have been developed as a centralized algorithm. 
	For instance, numerical methods are proposed in \cite{Yalcinoz1998,Gaing2003}. 
	For the distributed algorithm, there are only a few works which studied the EDP with power losses. 
	Authors of \cite{Binetti2014} propose a distributed algorithm considering the transmission losses. 
	However, it is assumed that the power losses can be computed at each iteration and the power losses were not dependent on the decision variables. 
	On the other hand, \cite{Zhao2017} considered the power losses which depend quadratically on the power generation of each generator. 
	However, it requires an initialization process which is not suitable for plug-and-play operation.
	
	In this work, we propose a continuous-time algorithm which solves the EDP with power losses in a distributed manner. 
	The EDP with power losses are formulated as a non-convex optimization problem \rt{with the assumption that the power losses are separable}. %
	\rt{ 
	Despite the non-convexity, it is shown that an optimal solution can be recovered using convex relaxation under mild assumptions.%
	}
	Proposed algorithm does not require any initialization process, thus allowing the plug-and-play operation. 
	In particular, the proposed algorithm is robust to changes such as change of demands or network topology. 
	\rt{
	The trade-off for having a robust algorithm is that the obtained solution is suboptimal. %
	However, it is shown that with sufficiently high gain, an optimal solution is recovered.%
	}
	Finally, behavior of the algorithm is analyzed when the problem is infeasible.%
		
	\textbf{Notation:} 
	For vectors $x_i \in \bR^{n}$ with $i=1,\ldots,N$, $[x_1^T,\ldots,x_N^T]^T$ is denoted by $[x_1;\cdots;x_N] \in \bR^{\bar{n}}$ where $\bar{n}:= Nn$. 
	An undirected graph is defined as $\cG=(\cN,\cE)$ where $\cN=\{1,\ldots,N\}$ is the node set and $\cE \subseteq \cN \times \cN$ is the edge set.
	The Laplacian matrix $L=[l_{ij}] \in {\mathbb{R}}^{N \times N}$ is defined as $l_{ij}:=-1$ if $(j,i)\in \cE$ and $l_{ij} := 0$ otherwise for $i \neq j$, and $l_{ii}:=-\sum_{j \neq i} l_{ij}$.
	Eigenvalues of $L$ is denoted as $ 0 =\sigma_1(L) \leq \ldots \leq \sigma_N(L)$.
	Given a set $\cX \subset \bR^n$, let $|x|_\cX := \inf_{z \in \cX} |x-z|$.
	We denote a set of continuously differentiable functions as $\cC^1$. 
	Derivative of a function $f(x)$ is denoted as $df/dx$ or $f'$.
	
	\section{Problem Formulation} \label{sec:prob}
	Consider the power network with $N$ nodes in the system. 
	Then, the EDP with power losses can be formulated as the following optimization problem:
	\begin{subequations} \label{eq:prim_prob}
		\begin{align}
		\min_{x_1,\ldots,x_N}~& \textstyle\sum_{i=1}^N f_i(x_i) \label{eq:prim_cost}\\
		\mathrm{subject\ to}~& \textstyle\sum_{i=1}^N d_i = \textstyle\sum_{i=1}^N  x_i - \phi(x) \label{eq:prim_eq}\\
		& x_i \in \cX_i \label{eq:prim_cap}, \quad \forall i \in \cN
		\end{align}
	\end{subequations}
	where $x_i \in \bR$ is the power generation (before losses) of the node $i$, $d_i \in \bR$ is the power demand, $f_i(x_i):\bR \rightarrow \bR$ is the local cost function, $x := [x_1;\ldots;x_N] \in \bR^N$, and $\phi(x):\bR^{N} \rightarrow \bR$ represents the power losses.
	The set $\cX_i := [\underline{x}_i, \bar{x}_i]$ is a nonempty closed interval where $\underline{x}_i$ and $\bar{x}_i$ are the minimum and maximum generation of the node $i$ respectively. 
	The objective of the problem \eqref{eq:prim_prob} is to minimize the generation cost \eqref{eq:prim_cost} subject to overall supply and demand balance considering the power losses \eqref{eq:prim_eq} and generation capacity constraints \eqref{eq:prim_cap}.

	Let $\cX:= \cX_1 \times \cdots \times \cX_N$, $\bar{x} := [\bar{x}_1;\ldots;\bar{x}_N]$, $d := [d_1;\ldots;d_N]$ and $\underline{x} := [\underline{x}_1;\ldots;\underline{x}_N]$.
	It is supposed that the information such as $f_i, x_i, d_i$, and $\cX_i$ is private to each node and is not shared with its neighbors. 
	
	In this work, we suppose that the loss is \textit{separable}, i.e., 
	\begin{align} \label{eq:loss}
	\phi(x) = \textstyle \sum_{i=1}^N \phi_i(x_i),
	\end{align}
	where $\phi_i(x_i):\bR \rightarrow \bR$ is a nonlinear function. 
	\rt{
	The loss given by \eqref{eq:loss} includes various forms of losses. 
	For example, it models losses of each generator such as copper losses or mechanical losses \cite{Grauers1996} where the separability assumption is naturally satisfied. 
	In addition, \eqref{eq:loss} also includes simplified model for the transmission losses. 
	Separable model for the transmission losses are also employed in previous works, e.g., see \cite{Yalcinoz1998} and \cite{Zhao2017}.
	}
	
	Note that the optimization problem \eqref{eq:prim_prob} is not a convex optimization problem due to the nonlinear equality constraint \eqref{eq:prim_eq}.
	
	\begin{asm} \label{ass:basic}
		The local cost function $f_i$ and loss function $\phi_i$ are $\cC^1$, $f_i$ is strictly convex and $\phi_i$ is convex over $\cX_i$ for $i=1,\dots,N$. 
		Moreover, $\phi_i$ satisfies 
		\begin{align} \label{ass:loss_slope}
		\frac{d \phi_i(x_i)}{d x_i} < 1, \quad \forall i \in \cN,
		\end{align}
		for all $\underline{x}_i \leq x_i \leq \bar{x}_i$.
	\end{asm}
	
	Inequality \eqref{ass:loss_slope} of Assumption \ref{ass:basic} implies that the incremental loss of each node cannot exceed the incremental generation, which is reasonable. 
	In what follows, we give a necessary and sufficient condition for the feasibility of the optimization problem \eqref{eq:prim_prob}.
	
	\begin{lem} \label{lem:feas}
		Suppose that Assumption \ref{ass:basic} holds. 
		Then, 
		\begin{align}\label{eq:feas_cond}
		\textstyle\sum_{i=1}^N \underline{x}_i - \phi_i(\underline{x}_i) \leq \textstyle\sum_{i=1}^{N} d_i \leq \textstyle\sum_{i=1}^{N} \bar{x}_i - \phi_i(\bar{x}_i)
		\end{align}
		holds if and only if the optimization problem \eqref{eq:prim_prob} is feasible.
	\end{lem}
	
	\begin{proof}
		Let $D(x) := \sum_{i=1}^{N} x_i - \phi_i(x_i) - d_i$. 
		Suppose that the problem \eqref{eq:prim_prob} is feasible. 
		Then, there exists a $z^* := [z_1^*;\ldots;z_N^*] \in \bR^N$ such that $z_i^* \in \cX_i$ and $D(z^*) = 0$. 
		Moreover, it follows from Assumption \ref{ass:basic} that
		$\partial D(x)/\partial x_i = 1 - d \phi_i(x_i)/d x_i > 0$
		for all $x_i \in \cX_i$. 
		Therefore, $D(x)$ is strictly increasing in each argument. 
		First, we will show that
		\begin{align} \label{eq:feas_upper}
		0 \leq \textstyle\sum_{i=1}^{N} \bar{x}_i - \phi_i(\bar{x}_i) - d_i
		\end{align}
		holds. 
		If $z^* = \bar{x}$, then \eqref{eq:feas_upper} holds with an equality. 
		If $z^* \neq \bar{x}$, then there exists an index $i\in \cN$ such that $\underline{x}_i \leq z_i^* < \bar{x}_i$ since $z^*$ is a feasible solution. 
		Therefore, $0=D(z^*) < D(\bar{x})$ holds since $D(\cdot)$ is strictly increasing in each argument, proving \eqref{eq:feas_upper}. 
		In a similar manner, it can be shown that $D(\underline{x}) \leq D(z^*)$ as well.
		
		Conversely, suppose \eqref{eq:feas_cond} holds. 
		Consider $D(\underline{x} + \alpha \cdot c)$, where $c \in [0,1]$ and $\alpha := [\bar{x}_1 - \underline{x}_1;\ldots;\bar{x}_N - \underline{x}_N] \in \bR^N$. 
		Then, $D(\underline{x} + \alpha \cdot c)$ is a continuous function.
		In addition, \eqref{eq:feas_cond} can be used to obtain
		$
		D(\underline{x}) \leq 0 \leq D(\bar{x}).
		$
		Therefore, it follows from the intermediate value theorem that there exists a $c^*\in [0,1]$ such that $D(\underline{x} + \alpha c^*) = 0$. 
		Hence, the problem \eqref{eq:prim_prob} is feasible with $\underline{x} + \alpha c^*$ as a  solution. 
	\end{proof}

	\begin{rem}
		It follows from the proof of Lemma \ref{lem:feas} that \eqref{eq:feas_cond} is a sufficient condition for feasibility regardless of \eqref{ass:loss_slope}.
		However, \eqref{eq:feas_cond} is not a necessary condition if \eqref{ass:loss_slope} does not hold. 
		For example, suppose
		$
		\sum_{i=1}^N \underline{x}_i - \phi_i(\underline{x}_i) > \sum_{i=1}^{N} d_i
		$
		such that \eqref{eq:feas_cond} does not hold. 
		Nevertheless, the problem \eqref{eq:prim_prob} may still be feasible if $d\phi_i(x_i)/dx_i > 1$. 
		In particular, if one \textit{loses} more power as $x_i$ is increased, a feasible solution may exist. 
		However, if such cases are not allowed (e.g., by assuming \eqref{ass:loss_slope}), then \eqref{eq:feas_cond} is indeed a necessary and sufficient condition for the feasibility.
	\end{rem}
		
	\section{A Centralized Solution} \label{sec:cent}
	In order to solve the non-convex optimization problem, the following assumption is made.
	
	\begin{asm} \label{ass:mono}
		$df_i(x_i)/dx_i > 0$ for all $x_i \in \cX_i$.
	\end{asm}
	
	Assumption \ref{ass:mono} is easily satisfied in practical scenarios. For instance, it is common to assume that the cost function is given by a quadratic function $f_i(x_i) = a_i + b_i x_i + c_i x_i^2$ where $b_i,c_i > 0$ and $\underline{x}_i \geq 0$. In such case, Assumption \ref{ass:mono} holds.
	
	The optimization problem \eqref{eq:prim_prob} will be relaxed into the following convex optimization problem:
	\begin{subequations} \label{eq:prim_prob_relax}
		\begin{align}
		\min_{x_1,\ldots,x_N} ~& \textstyle\sum_{i=1}^N f_i(x_i)\\
		\mathrm{subject~to} ~&\textstyle\sum_{i=1}^N d_i - x_i + \phi_i(x_i)  \leq 0 \label{eq:prim_eq_relax}\\
		& x_i \in \cX_i, \quad \forall i \in \cN
		\end{align}
	\end{subequations}
	which will be called as the \textit{relaxed problem}. 
	Note that the relaxed problem \eqref{eq:prim_prob_relax} is a convex optimization problem since the equality constraint \eqref{eq:prim_eq} is relaxed into an inequality constraint as \eqref{eq:prim_eq_relax}.
	Nevertheless, it will be shown that an optimal solution of \eqref{eq:prim_prob} is obtained by solving \eqref{eq:prim_prob_relax}.
	
	For the optimization problem \eqref{eq:prim_prob_relax}, define the Lagrangian function $\cL^{r}(x,\lambda):\cX \times \bR \rightarrow \bR$ as
	\begin{align*}
	\cL^r(x,\lambda) 
	{=} \sum_{i=1}^N f_i(x_i) + \lambda \left( d_i - x_i + \phi_i(x_i) \right)
	{=:} \sum_{i=1}^N \cL^r_i(x_i,\lambda)
	\end{align*}
	where $\lambda \in \bR$ is the dual variable. 
	Then, the dual function $g^r(\lambda)$ for \eqref{eq:prim_prob_relax} can be written as
	\begin{align*}
	& g^r(\lambda) = \min_{x \in \cX}\cL^r(x,\lambda) = \sum_{i=1}^N \min_{x_i \in \cX_i} \cL_i^r(x_i,\lambda) =: \sum_{i=1}^N g^r_i(\lambda).
	\end{align*}
	The following lemma gives the expression for $g_i^r(\lambda)$.
	
	\begin{lem} \label{lem:lag_opt}
		Suppose that Assumptions \ref{ass:basic} and \ref{ass:mono} hold.
		Let $v_i:\cX_i \rightarrow \bR$ be defined as 
		\begin{align} \label{eq:lambda_first}
		v_i(x_i) := \frac{df_i(x_i)}{dx_i} \cdot \left(1 - \frac{d\phi_i(x_i)}{dx_i}\right)^{-1}.
		\end{align}
		Then, $v_i(x_i)$ is a strictly increasing function for $\underline{x}_i \leq x_i \leq \bar{x}_i$. 
		Moreover, suppose $\lambda \geq 0$ and let $\hat{x}_i(\lambda)$ be  
		$$
		\hat{x}_i(\lambda) := \begin{cases}
		\underline{x}_i & 0 \leq \lambda \leq v_i(\underline{x}_i) \\
		v_i^{-1}(\lambda) & v_i(\underline{x}_i) < \lambda < v_i(\bar{x}_i) \\ 
		\bar{x}_i & v_i(\bar{x}_i) \leq \lambda.
		\end{cases}
		$$
		Then, $\hat{x}_i(\lambda)$ is the unique minimizer of $\cL^r_i(x_i,\lambda)$, i.e., 
		$
		\hat{x}_i(\lambda) = \argmin_{x_i \in \cX_i} \cL^r_i(x_i,\lambda).
		$
	\end{lem}
	
	\begin{proof}
		In order to show $v_i(x_i)$ is a strictly increasing function for $\underline{x}_i \leq x_i \leq \bar{x}_i$, let $z_1,z_2 \in \cX_i$ such that $z_1 > z_2$.
		Then, it follows from \eqref{eq:lambda_first} that
		\begin{align*}
			v_i(z_1) - v_i(z_2) = \frac{f_i'(z_1)\left(1-\phi_i'(z_2)\right) - f_i'(z_2)\left(1-\phi_i'(z_1)\right)}{\left(1-\phi_i'(z_1)\right) \cdot \left(1-\phi_i'(z_2)\right)}.
		\end{align*}
		Since $0 <1-\phi_i'(z_k)$ holds for $k=1,2$ by \eqref{ass:loss_slope}, it is equivalent to show the strict positivity of 
		\begin{align} \label{eq:positive}
			f_i'(z_1)\left(1-\phi_i'(z_2)\right) - f_i'(z_2)\left(1-\phi_i'(z_1)\right).
		\end{align}
		From Assumptions \ref{ass:basic} and \ref{ass:mono}, it holds that
		\begin{align*}
			\frac{f_i'(z_1)}{f'_i(z_2)} > 1 \geq \frac{1-\phi_i'(z_1)}{1-\phi_i'(z_2)}
		\end{align*}
		which proves $v_i(x_i)$ is strictly increasing.
		Consequently, $v_i^{-1}(\lambda)$ is well-defined for $v_i(\underline{x}_i) \leq \lambda \leq v_i(\bar{x}_i)$.
		
		Next, suppose $\lambda \geq 0$. Then, it follows that 
		$\cL^r_i(x_i,\lambda)$ is a strictly convex function in $x_i$ for any fixed $\lambda \geq 0$. Hence, it has a unique minimum and the minimum of $\cL_i^r(x_i,\lambda)$ is obtained when
		$$
		\frac{\partial\cL^r_i(x_i,\lambda)}{\partial x_i} = \frac{df_i(x_i)}{dx_i} + \lambda \frac{d\phi_i(x_i)}{dx_i} - \lambda  = 0
		$$
		which is equivalent to $v_i(x_i) = \lambda$.
		Therefore, $\arg\!\min_{x_i \in \cX_i} \cL^r_i(x_i,\lambda)= v_i^{-1}(\lambda)$ holds for
		$
		\underline{x}_i \leq v_i^{-1}(\lambda)\leq \bar{x}_i
		$
		which becomes
		$
		v_i(\underline{x}_i) \leq \lambda \leq v_i(\bar{x}_i)
		$
		since $v_i(x_i)$ is strictly increasing.
		
		Finally, let $\lambda$ be a fixed scalar such that $\lambda > v_i(\bar{x}_i)$ holds. Then for any $\underline{x}_i \leq x_i < \bar{x}_i$, it follows that 
		\begin{align*}
			\frac{\partial \cL_i^r(x_i,\lambda)}{\partial x_i}  
			< f'_i(\bar{x}_i) - \lambda + \lambda \phi_i'(\bar{x}_i)
			& < 0
		\end{align*}
		where the first strict inequality follows since $f_i$ is strictly convex. 
		Since $\cL_i^r(x_i,\lambda)$ is strictly decreasing, its minimum is obtained at $x_i = \bar{x}_i$. 
		The case when $0 \leq \lambda < v_i(\underline{x}_i)$ can be proven in a similar manner.
	\end{proof}
	
	\begin{rem}
		If $p_i^*:=\bar{x}_i = \underline{x}_i$, (e.g., node $i$ generates a fixed amount of power or no power at all if $p_i^* = \bar{x}_i=\underline{x}_i=0$), then $\hat{x}_i(\lambda) = p_i^*$. This is consistent with Lemma \ref{lem:lag_opt}. In particular, $\hat{x}_i(\lambda) = \argmin_{x_i \in \cX_i} \cL^r_i(x_i,\lambda)$ holds. 
	\end{rem}
		
	The dual problem of the relaxed problem \eqref{eq:prim_prob_relax} becomes\begin{subequations} \label{eq:prim_relax_dual}
		\begin{align}
			\max_{\lambda \in \bR}~&  g^r(\lambda) = \textstyle\sum_{i=1}^N g_i^r(\lambda)\\
			\mathrm{subject~to~}~& \lambda \geq 0.
		\end{align}
	\end{subequations}
	where $\lambda$ is constrained to be non-negative since \eqref{eq:prim_eq_relax} is an inequality constraint.
	
	In order to solve the constrained optimization problem such as \eqref{eq:prim_relax_dual} using continuous-time algorithms, discontinuous vector fields are often employed \cite{Yi2016,Cherukuri2016a} to constrain variables to the feasible set. 
	However, such methods are harder to implement due to the discontinuity. 
	Different from these approaches, we simply extend the domain of $\hat{x}_i(\lambda)$ as below:
	\begin{align} \label{eq:optimal_power_extend}
		\hat{x}_i(\lambda) := \begin{cases}
		\underline{x}_i &  \lambda \leq v_i(\underline{x}_i) \\
		v_i^{-1}(\lambda) & v_i(\underline{x}_i) < \lambda < v_i(\bar{x}) \\ 
		\bar{x}_i & v_i(\bar{x}_i) \leq \lambda
		\end{cases}
	\end{align}
	where we have defined $\hat{x}_i(\lambda) = \underline{x}_i$ for all $\lambda < 0$.
	Accordingly, define the \textit{modified dual function} as $g^m(\lambda) := \sum_{i=1}^N\cL_i^r(\hat{x}_i(\lambda),\lambda)$. Then, the following result holds.
	\begin{lem} \label{lem:modified_dual_prop}
		Suppose that Assumptions \ref{ass:basic} and \ref{ass:mono} hold.
		Then, the modified dual function given by 
		$$
		g^m(\lambda) = \textstyle\sum_{i=1}^{N} \cL_i^r(\hat{x}_i(\lambda),\lambda) =: \textstyle\sum_{i=1}^{N} g_i^m(\lambda)
		$$
		is $\cC^1$ and concave for all $\lambda \in \bR$.
	\end{lem}
	
	\begin{proof}
		For $\lambda > 0$, it follows from Lemma \ref{lem:lag_opt} that 
		$\hat{x}_i(\lambda) = \argmin_{x_i \in \cX_i} \cL^r_i(x_i,\lambda).$
		Therefore, \cite[Prop.7.1.1]{bert2016} states that $g^r(\lambda)$ is concave and $\cC^1$. Since $g^m(\lambda) = g^r(\lambda)$ for $\lambda > 0$, $g^m(\lambda)$ is also concave and $\cC^1$ for $\lambda > 0$.
		For $\lambda < 0$, it follows from the definition of $\hat{x}_i(\lambda)$ that
		$$
		g^m(\lambda) = \textstyle\sum_{i=1}^{N} f_i(\underline{x}_i) + \lambda(d_i + \phi_i(\underline{x}_i) - \underline{x}_i).
		$$
		Thus, it is obvious that $g^m(\lambda) \in \cC^1$ for $\lambda < 0$. 
		Moreover, it can be also verified that $g^m(\lambda)$ is differentiable at $\lambda = 0$.
		It is left to show $g^m(\lambda)$ is concave for $\lambda \leq 0$.
		However, this follows directly since $g^m(\lambda) \in \cC^1$ and it is a linear function of $\lambda$ for $\lambda \leq 0$.
	\end{proof}

	Now, instead of \eqref{eq:prim_relax_dual}, consider the unconstrained optimization problem
	\begin{align} \label{eq:prim_relax_dual_modified}
		\max_{\lambda \in \bR} g^m(\lambda) = \textstyle\sum_{i=1}^N g_i^m(\lambda)
	\end{align}
	with the gradient ascent algorithm given by 
	\begin{align} \label{eq:cen_dyn}
		\dot{\lambda} = \frac{dg^m(\lambda)}{d\lambda} = {\sum\nolimits_{i=1}^N} \frac{d g^m_i(\lambda)}{d \lambda}.
	\end{align}
	Using the result of \cite[Prop.7.1.1]{bert2016}, it can be verified that the derivative of $g_i^m(\lambda)$ becomes
	$$
	\frac{dg_i^m(\lambda)}{d\lambda} = d_i - \hat{x}_i(\lambda) + \phi_i(\hat{x}_i(\lambda)), \quad \forall \,i=1,\ldots,N.
	$$
	Note it follows from Assumption \ref{ass:basic} and \eqref{eq:optimal_power_extend} that $dg_i^m(\lambda)/d\lambda$ (and hence $dg^m(\lambda)/d\lambda$) is monotonically decreasing, uniformly bounded and uniformly continuous.

	Before presenting the centralized solution, recall the following result customized from \cite[Prop. 6.1.5]{bert2016}.
	
	\begin{lem}  \label{lem:opt_cond}
		Let $x^* := [x_1^*;\ldots;x_N^*] \in \bR^N$ and $\lambda^* \in \bR$. Then, the pair $(x^*,\lambda^*)$ satisfies
		\begin{subequations}
			\begin{gather*}
				x_i^*  \in \cX_i, \quad \lambda^* \geq 0, \quad \textstyle\sum_{i=1}^N d_i - x_i^* + \phi_i(x_i^*) \leq 0,\\
				\lambda^* \left( \textstyle\sum_{i=1}^N d_i - x_i^* + \phi_i(x_i^*) \right)  = 0, \\
				x^*  \in \argminl_{x \in \cX} \cL^r(x,\lambda^*), 
			\end{gather*}        
		\end{subequations}
		if and only if $(x^*,\lambda^*)$ is an optimal solution-geometric multiplier pair of \eqref{eq:prim_prob_relax}.
	\end{lem}
	
	Now, it will be shown that an optimal solution of \eqref{eq:prim_prob} can be obtained from \eqref{eq:cen_dyn}.
	
	\begin{thm} \label{thm:opt}
		Suppose that Assumptions \ref{ass:basic} and \ref{ass:mono} hold and that the optimization problem \eqref{eq:prim_prob} is feasible (i.e., \eqref{eq:feas_cond} holds). Consider the gradient ascent algorithm given by \eqref{eq:cen_dyn}.
		Then, $\lim_{t \rightarrow \infty}\lambda(t) = \lambda^*$ where $\lambda^*$ is an optimal solution of \eqref{eq:prim_relax_dual_modified}.
		Moreover, $\hat{x}_i(\lambda^*)$ is an optimal solution of the optimization problem \eqref{eq:prim_prob}.
	\end{thm}
	\begin{proof}
		From Lemma \ref{lem:modified_dual_prop}, it follows that \eqref{eq:cen_dyn} is a gradient ascent algorithm for the concave function $g^m(\lambda)$.
		Hence, it can be easily shown that $\lambda(t)$ converges to an optimal solution of \eqref{eq:prim_relax_dual_modified} using \eqref{eq:feas_cond} and \eqref{eq:optimal_power_extend}.
		
		It is left to show $\hat{x}(\lambda^*) := [\hat{x}_1(\lambda^*);\ldots;\hat{x}_N(\lambda^*)] \in \bR^N$ is an optimal solution to the problem \eqref{eq:prim_prob}. 
		For this, it will be shown that the pair $(\hat{x}(\lambda^*),\lambda^*)$ satisfies the optimality conditions for the problem \eqref{eq:prim_prob_relax} provided in Lemma \ref{lem:opt_cond} while inequality constraint \eqref{eq:prim_eq_relax} is satisfied with an equality.
		
		From the first order optimality condition for \eqref{eq:prim_relax_dual_modified}, it follows that 
		\begin{align} \label{eq:thm1_feasibility}
			\textstyle\sum_{i=1}^{N} d_i - \hat{x}_i(\lambda^*) + \phi_i(\hat{x}_i(\lambda^*)) = 0.
		\end{align}
		Moreover, $\hat{x}_i(\lambda^*) \in \cX_i$ by the definition. 
		Therefore, $\hat{x}_i(\lambda^*)$ is a feasible solution to \eqref{eq:prim_prob_relax}. 
		Now, consider the case when $\lambda^* \geq 0$. Then, it follows from Lemma \ref{lem:lag_opt} that $\hat{x}_i(\lambda^*) = \argmin_{x_i\in\cX_i} \cL^r_i(x_i,\lambda^*)$. 
		Therefore, $\hat{x}_i(\lambda^*)$ is an optimal solution of \eqref{eq:prim_prob_relax} due to Lemma \ref{lem:opt_cond}. 
		If $\lambda^* < 0$, then it follows from \eqref{eq:optimal_power_extend} that
		$$
		\hat{x}_i(\lambda^*) = \hat{x}_i(0) = \textstyle\argmin_{x_i \in \cX_i} \cL^r(x_i,0).
		$$
		Hence, we can conclude $(\hat{x}(0),0)$ is an optimal solution-geometric multiplier pair of \eqref{eq:prim_prob_relax} using Lemma \ref{lem:opt_cond}.
		Consequently, $\hat{x}_i(\lambda^*)$ is an optimal solution to \eqref{eq:prim_prob_relax}. 
		
		Finally, it follows from \eqref{eq:thm1_feasibility} that $\hat{x}_i(\lambda^*)$ satisfies constraint \eqref{eq:prim_eq_relax} with an equality. Therefore, $\hat{x}_i(\lambda^*)$ is an optimal solution of the problem \eqref{eq:prim_prob}.
	\end{proof}
	
	From Theorem \ref{thm:opt}, it can be seen that the optimization problem can be solved using \eqref{eq:cen_dyn}. In particular, the optimal generation for each node is obtained using \eqref{eq:optimal_power_extend}.
	
	\section{A Distributed Solution} \label{sec:dist}
	In this section, a distributed algorithm for solving the dual problem \eqref{eq:prim_relax_dual_modified} (and hence the primal problem \eqref{eq:prim_prob}) is proposed.
	Suppose that the node $i$ runs 
	\begin{subequations} \label{eq:dist_dyn}
		\begin{align} 
			\dot{\lambda}_i(t) & = \frac{d g^m_i}{d\lambda}(\lambda_i(t)) + k \sum_{j \in \cN_i} (\lambda_j(t) - \lambda_i(t)) \label{eq:dist_dyn_algo}\\
			x_i(t) & = \hat{x}_i(\lambda_i(t)) \label{eq:dist_dyn_power}
		\end{align}
	\end{subequations}
	where $\lambda_i \in \bR$ is the estimate of the dual variable by the node $i$, $x_i(t)$ is the power generation of the  node $i$ (at time $t$), $k > 0$ is the coupling gain and $\cN_i := \{j \in \cN \mid (j,i) \in \cE \}$.
	The proposed algorithm \eqref{eq:dist_dyn} is an extension of \cite{Yun2019} to the EDP with power losses.
	For the distributed algorithm, we make the following assumption.
	\begin{asm} \label{ass:graph}
		Communication graph is undirected and connected.
	\end{asm}
	
	Note \eqref{eq:dist_dyn} is a distributed algorithm as $dg^m_i/d\lambda$ can be computed by the node $i$ only using the local information.
	Moreover, only the estimate of the dual variable is communicated between agents and no private information such as $d_i$ or $f_i(\cdot)$ are exchanged.
	
	Let $\bm{\lambda} := [\lambda_1;\ldots;\lambda_N] \in \bR^N$ be the stack of $\lambda_i$ and $G(\bm{\lambda}) := [dg^m_1(\lambda_1)/d\lambda;\ldots;dg^m_N(\lambda_N)/d\lambda]$. 
	Then \eqref{eq:dist_dyn_algo} can be written as 
	\begin{equation} \label{eq:thm_dyn}
		\dot{\bm{\lambda}} = G(\bm{\lambda}) - kL\bm{\lambda}
	\end{equation}
	where $L \in \bR^{N \times N}$ is a symmetric Laplacian matrix. Denoting $1_N := [1;\ldots;1] \in \bR^N$, it follows that there exists a matrix $W = [(1/N) 1_N^T; R^T] \in \bR^{N \times N}$ and $W^{-1} = [1_N, Q]$ such that $WLW^{-1} = \mathrm{diag}(0,\sigma_2(L),\ldots,\sigma_N(L))$ where $R \in \bR^{N \times (N-1)}$ and $Q \in \bR^{N \times (N-1)}$.
	Moreover, it can be checked that $|Q| = \sqrt{N}$ and $|R| = 1/\sqrt{N}$ \cite{Kim2016}. 
	
	Now apply the following coordinate transformation
	\begin{align} \label{eq:coord_trans}
		\xi := \begin{bmatrix}
		\bar{\xi} \\ \tilde{\xi}
		\end{bmatrix}
		=
		W\bm{\lambda} = \begin{bmatrix}
		\frac{1}{N} 1_N^T \\ R^T
		\end{bmatrix}
		\bm{\lambda}
	\end{align}
	where $\bar{\xi} \in \bR$ and $\tilde{\xi} \in \bR^{N-1}$. 
	In addition, it follows that $\bm{\lambda} = W^{-1}\xi$, or $\lambda_i = \bar{\xi} + Q_i \tilde{\xi}$ where $Q_i$ is the $i$-th row of $Q$. Then the system \eqref{eq:thm_dyn} is transformed into
	\begin{subequations} \label{eq:dyn_trans}
		\begin{align} 
			\dot{\bar{\xi}} & = \frac{1}{N} \sum_{i=1}^N \frac{dg_i^m}{d\lambda}(\bar{\xi}) + \frac{1}{N} \tilde{g}\big(\bar{\xi},\tilde{\xi}\big), \label{eq:dyn_trans_1}  \\
			\dot{\tilde{\xi}} & = -kR^TLQ \tilde{\xi}+ R^T G\left(1_N\bar{\xi} + Q\tilde{\xi}\,\right), \label{eq:dyn_trans_2}
		\end{align} 
	\end{subequations}
	where $\xi(0) = W \bm{\lambda}(0)$ and $\tilde{g}(\bar{\xi},\tilde{\xi}) := \sum_{i=1}^N ({dg_i^m}/{d\lambda})(\bar{\xi} + Q_i \tilde{\xi}) - ({dg_i^m}/{d\lambda})(\bar{\xi})$.
	Convergence of the proposed algorithm is stated below.
	
	\begin{thm} \label{thm:feas}
		Consider the distributed algorithm \eqref{eq:dist_dyn}.
		Suppose that Assumptions \ref{ass:basic}, \ref{ass:mono}, and \ref{ass:graph} hold. 
		Also assume that the optimization problem \eqref{eq:prim_prob} is feasible (i.e., \eqref{eq:feas_cond} holds).
		Then, for any $k > 0$, the solution of \eqref{eq:dist_dyn} converges to a point and satisfies
		\begin{align*}
			\lim\limits_{t \rightarrow \infty} \textstyle\sum_{i=1}^N x_i(\lambda_i(t))  - \phi_i\big(x_i(\lambda_i(t)) \big) = \textstyle\sum_{i=1}^N d_i
		\end{align*}
		for any initial conditions $\lambda_i(0) \in \bR$.
	\end{thm}
	
	\begin{proof}
		From \eqref{eq:dyn_trans_2}, it follows that $\tilde{\xi}(t)$ is bounded since $R^TLQ$ is positive definite and $G( \cdot )$ is bounded.
		It can also be verified that $\bar{\xi}(t)$ is bounded using \eqref{eq:feas_cond} and \eqref{eq:optimal_power_extend}.
		Since \eqref{eq:coord_trans} is a linear transformation, it follows that the solution of \eqref{eq:dist_dyn} is bounded.
		
		Now, let $V(\bm{\lambda}) = -\sum_{i=1}^{N} g^m_i(\lambda_i) + (k/2) \bm{\lambda}^T L \bm{\lambda}$ be a candidate function.
		Then its time derivative becomes
		\begin{align*}
			\dot{V} =  \left( - G(\bm{\lambda})+ k L \bm{\lambda} \right)^T \dot{\bm{\lambda}}
			= -\left|G(\bm{\lambda}) - kL\bm{\lambda}\right|^2 \leq 0.
		\end{align*}
		Thus, LaSalle's invariance principle can be applied to conclude that $\lambda_i(t)$ approaches to the set $E := \{\bm{\lambda} \mid \dot{V}(\bm{\lambda})=0\}$.
		However, note that the set $E$ is the set of equilibrium points of \eqref{eq:thm_dyn}.
		Thus, convergence to a point can be obtained by applying Lemma A.3 from \cite{Cherukuri2017}.
		
		For the feasibility of the converged solution, let $\hat{\lambda} = [\hat{\lambda}_1;\ldots;\hat{\lambda}_N]:= \lim_{t \rightarrow \infty} \bm{\lambda}(t) \in E$.
		Then, it holds that
		$
		G(\hat{\lambda}) - kL\hat{\lambda} = 0.
		$
		Multiplying $1_N^T$ from the left, we obtain
		\begin{align*}
			1_N^T G(\hat{\lambda}) = \textstyle\sum_{i=1}^N d_i - \hat{x}_i(\hat{\lambda}_i) + \phi_i\big(\hat{x}_i(\hat{\lambda}_i)\big)  = 0
		\end{align*}
		since $1_N^TL=0$.
	\end{proof}
	
	Result of Theorem \ref{thm:opt} states that for any $k>0$, the algorithm \eqref{eq:dist_dyn} converges to a \textit{feasible} solution of \eqref{eq:prim_prob}.
	However, the optimality of the converged solution has not been stated.
	In what follows, it is shown that the optimality can be recovered using high coupling gain.
	\begin{thm} \label{thm:dist_opt}
		Consider the distributed algorithm \eqref{eq:dist_dyn} and suppose that the assumptions of Theorem \ref{thm:feas} hold.
		Then, for any $\epsilon >0$, there exists $\bar{k} >0$ and a function $T(\bm{\lambda}(0),k)$ such that for all $k > \bar{k}$, it holds that 
		\begin{align*}
			|\hat{x}_i(\lambda_i(t)) - \hat{x}_i(\lambda^*)| \leq \epsilon, \quad \forall i \in \cN, \quad \forall t \geq T(\bm{\lambda}(0),k),
		\end{align*}
		where $\hat{x}_i(\lambda^*)$ is an optimal solution of \eqref{eq:prim_prob}.
	\end{thm}
	\begin{proof}
		Let $p_i(z) := z - \phi_i(z)$ for all $z \in \cX_i$. 
		Then, $p_i(z)$ is uniformly continuous and strictly increasing.
		Therefore, $p_i^{-1}(\cdot)$ is also uniformly continuous and strictly increasing.
		Thus, there exists $\delta_1> 0$ such that for all $i \in \cN$,
		\begin{align*}
			|p_i(a') - p_i(b')| \leq \delta_1 \implies |a' - b'| \leq \epsilon
		\end{align*}
		and $\delta_2>0$ such that
		\begin{align*}
			|a - b| \leq \delta_2 \implies \left|\frac{dg_i^m(a)}{d\lambda} - \frac{dg_i^m(b)}{d\lambda}\right| \leq \frac{\delta_1}{3N}.
		\end{align*}
		Also define $M := \max_{\lambda} |G(\lambda)|$ which exists since $G(\lambda)$ is bounded.
		Finally, define $\bar{k} := 2 M/\sigma_2(L)\delta_2$.
		
		Let $V(\tilde{\xi}) = (1/2) \tilde{\xi}^T \tilde{\xi}$ be a candidate Lyapunov function. 
		Then, the time derivative of $V$ along the trajectories of \eqref{eq:dyn_trans_2} becomes
		\begin{align*}
			\dot{V} & \leq -k \sigma_2(L) |\tilde{\xi}|^2 + |R^T||G(1_N\bar{\xi} + Q \tilde{\xi})||\tilde{\xi}| \\ 
			& \leq -\frac{k \sigma_2(L)}{2}|\tilde{\xi}|^2, \quad \forall |\tilde{\xi}| \geq \frac{2 M}{k \sigma_2(L)\sqrt{N}}
		\end{align*}
		where $\sigma_2(L)$ is the second smallest eigenvalue of the $L$.
		Therefore, for all $k \geq \bar{k}$, it holds that
		$$
		|\tilde{\xi}(t)| \leq \frac{2 M }{k \sigma_2(L) \sqrt{N}} \leq \frac{\delta_2}{\sqrt{N}}
		$$
		for all $t \geq T_1(\bm{\lambda}(0),k)$ where 
		\begin{align*}
			T_1(\bm{\lambda}(0),k) = \ln\left( \frac{\sqrt{N}|R^T\bm{\lambda}(0)|}{\delta_2} \right) \cdot \frac{2}{k \sigma_2(L)} ,
		\end{align*}
		and $T_1(\bm{\lambda}(0),k) := 0$ if $|\tilde{\xi}(0)| \leq \delta_2/\sqrt{N}$ (or equivalently $|R^T\bm{\lambda}(0)| \leq \delta_2/\sqrt{N}$).
		Thus, $|Q_i\tilde{\xi}(t)| \leq \sqrt{N} \cdot (\delta_2/\sqrt{N}) = \delta_2.$  
		Hence, we obtain
		\begin{align} \label{eq:delta_eps}
			\left| \sum_{i=1}^N \frac{dg^m_i(\bar{\xi} + Q_i\tilde{\xi})}{d\lambda} - \frac{dg^m_i(\bar{\xi})}{d\lambda} \right| \leq N \cdot \frac{\delta_1}{3N} = \frac{\delta_1}{3}.
		\end{align}
		
		Define $\Lambda_{\delta_1}^* := \{w\in \bR \mid |dg^m(w)/d\lambda| \leq 2\delta_1/3\}$ and denote it as $\Lambda^*_{\delta_1} = [\underline{\Lambda}, \bar{\Lambda}]$.
		In fact, since $(dg^m/d\lambda)(\lambda)$ is monotonically decreasing, it holds that $(dg^m/d\lambda)(\lambda) \geq 2\delta_1/3$ for $\lambda \leq \underline{\Lambda}$ and $(dg^m/d\lambda)(\lambda) \leq -2\delta_1/3$ for $\lambda \geq \bar{\Lambda}$.
		From here on, suppose that $t \geq T_1(\bm{\lambda}(0),k)$.
		Then, the trajectory of $\bar{\xi}(t)$ approaches to $\Lambda^*_{\delta_1}$ with the speed of at least $\delta_1/(3N)$.
		For instance, if $\bar{\xi}(t) \geq \bar{\Lambda}$,  it follows from \eqref{eq:dyn_trans_1} and \eqref{eq:delta_eps} that 
		$$
		\dot{\bar{\xi}} \leq -\frac{2\delta_1}{3N} + \frac{\delta_1}{3N} = -\frac{\delta_1}{3N}.
		$$ 
		Hence, $\bar{\xi}(t)$ converges to $\Lambda^*_{\delta_1}$ in finite time.
		In order to compute the convergence time, it follows from \eqref{eq:feas_cond} that 
		\begin{align*}
			|\dot{\bar{\xi}}(t)| & \leq \left|\frac{1}{N}\sum_{i=1}^N d_i - x_i(\lambda_i(t)) + \phi_i\big(x_i(\lambda_i(t))\big) \right|\\
			& \leq \frac{1}{N} \sum_{i=1}^N \bar{x}_i - \phi_i(\bar{x}_i) - \underline{x}_i + \phi_i(\underline{x}_i) =: \frac{\Delta}{N}.
		\end{align*}
		Therefore, it follows that
		$$
		|\bar{\xi}\big(T_1(\bm{\lambda}(0),k)\big)| \leq |\bar{\xi}(0)| + \frac{\Delta}{N}T_1(\bm{\lambda}(0),k) =: \zeta^*.
		$$ 
		Define $D^*$ as
		$
		D^* := \max\big(|\zeta^*|_{\Lambda^*_{\delta_{1}}},\left|-\zeta^*\right|_{\Lambda^*_{\delta_1}}\big).
		$
		Then, it holds that 
		$
		\bar{\xi}(t) \in \Lambda^*_{\delta_1} 
		$
		for all $ t \geq T(\bm{\lambda}(0),k)$, where
		\begin{align*}
			T(\bm{\lambda}(0),k) :=  \frac{3ND^*}{\delta_1} + T_1(\bm{\lambda}(0),k).
		\end{align*}
		Finally, for some fixed $i$ and for all $t \geq T(\bm{\lambda}(0),k)$,
		\begin{align}
			& \left| p_i\big(\hat{x}_i(\lambda_i(t))\big) - p_i\big(\hat{x}_i(\lambda^*)\big)\right| \nonumber \\ 
			& \leq  \sum_{j=1}^N\left| p_j\big(\hat{x}_j(\lambda_i(t)\big) - p_j\big(\hat{x}_j(\lambda^*)\big) \right| \nonumber \label{eq:thm3_error_1}\\ 
			& =  \left| \sum_{j=1}^N p_j\big(\hat{x}_j(\lambda_i(t))\big) - p_j\big(\hat{x}_j(\lambda^*)\big) \right| = \left| -\frac{dg^m}{d\lambda}(\lambda_i) \right|
		\end{align}
		where the first equality holds since $p_j(\hat{x}_j(\cdot))$ is an increasing function, and the second equality holds since $\sum_{j=1}^N p_j\big(\hat{x}_j(\lambda^*)\big) = \sum_{j=1}^N d_j$.
		Note, we have
		\begin{align} \label{eq:g_m_bound}
			\left| \frac{dg^m}{d\lambda}(\lambda_i) \right| & \leq \left| \frac{dg^m}{d\lambda}(\bar{\xi}) \right| + \left| \frac{dg^m}{d\lambda}(\bar{\xi} + Q_i \tilde{\xi}) - \frac{dg^m}{d\lambda}(\bar{\xi}) \right| \nonumber \\
			& \leq \frac{2\delta_1}{3} + \frac{\delta_1}{3} = \delta_1.
		\end{align}
		Therefore, \eqref{eq:thm3_error_1} and \eqref{eq:g_m_bound} implies
		\begin{align*}
			\left|p_i\big(\hat{x}_i(\lambda_i(t))\big) - p_i(\hat{x}_i(\lambda^*)) \right| \leq \delta_1.
		\end{align*}
		By the definition of $\delta_1$, it follows that
		$$
		\left| \hat{x}_i(\lambda_i(t)) - \hat{x}_i(\lambda^*) \right| \leq \epsilon
		$$
		for all $t \geq T(\bm{\lambda}(0),k)$.
	\end{proof}
	
	\rt{
		An important feature of the proposed algorithm \eqref{eq:dist_dyn} is that it is an initialization-free algorithm and hence allows plug-and-play operation. 
		In particular, it can be seen from Theorem \ref{thm:feas} that the proposed algorithm converges to a feasible solution \textit{regardless} of the initial condition.
		Therefore, even if some parameters of the optimization problem \eqref{eq:prim_prob} changes, the solution of \eqref{eq:dist_dyn} converges to a feasible solution of the new problem.
		Additionally, converged solution is close to an optimal if the coupling gain $k$ is chosen as stated in Theorem \ref{thm:dist_opt}.
		To bound the performance uniformly across changes, the coupling gain must be chosen sufficiently large to incorporate all possible cases.
		For instance, such gain can be found from the worst case scenario by assuming that the network has a known maximum capacity and parameters such as $f_i$ and $\phi_i$ are from a finite collection.
		More detailed discussions can be found in \cite[Sec. 6.1]{Yun2019} or \cite{Kim2016a}.%
	}
	
	\begin{rem}
		Using similar arguments as in \eqref{eq:g_m_bound}, it holds that the constraint violation and the objective error $f_i(x_i(t)) - f_i(\hat{x}_i(\lambda^*))$ become small in finite time. 
		In particular, $|\sum_{i=1}^N d_i - x_i(t) + \phi_i(x_i(t))| \leq \delta_1$ and 
		$
		|f_i\big(x_i(\lambda_i(t))\big) - f_i(\hat{x}_i(\lambda^*))| \leq c_i
		$
		holds for all $t \geq T(\bm{\lambda}(0),k)$ where $c_i > 0$ is such that $|a - b| \leq \epsilon$ implies $|f_i(a) - f_i(b)| \leq c_i$ for any $a,b \in \cX_i$.
		Thus, constraint violation and objective error can be made arbitrarily small by reducing $\epsilon$ which leads to higher coupling gain $k$.
	\end{rem}
	
	Results of Theorems \ref{thm:feas} and Theorem \ref{thm:dist_opt} assume that the problem \eqref{eq:prim_prob} is feasible. 
	The behavior of the proposed algorithm is also analyzed when the problem is infeasible.
	
	\begin{thm} \label{thm:infeas}
		Suppose that Assumptions \ref{ass:basic}, \ref{ass:mono} and \ref{ass:graph} hold.
		Assume that the optimization problem \eqref{eq:prim_prob} is infeasible. 
		Specifically, suppose that 
		$
		\textstyle\sum_{i=1}^{N} d_i - \bar{x}_i + \phi_i(\bar{x}_i) > 0
		$
		holds.
		Then, for all $i \in \cN$, $\lambda_i(t)$ diverges to $+\infty$ and
		$
		\lim_{t \rightarrow \infty} \dot{\lambda}_i(t) = D_0,
		$
		where $D_0 := (\sum_{i=1}^N d_i - \bar{x}_i + \phi_i(\bar{x}_i))/N > 0$. Similar result also holds in the case of $
		\textstyle\sum_{i=1}^{N} d_i - \underline{x}_i + \phi_i(\underline{x}_i) < 0.
		$
	\end{thm}
	\begin{proof}
		 From the proof of Theorem \ref{thm:feas}, it holds that $\tilde{\xi}(t)$ is bounded. Moreover, we have
		 \begin{align*}
			 \dot{\bar{\xi}} & = \frac{1}{N} \sum_{i=1}^N d_i - \hat{x}_i(\bar{\xi} + Q_i \tilde{\xi})  + \phi_i\left(\hat{x}_i( \bar{\xi} + Q_i \tilde{\xi})\right) \\ 
			 & \geq \frac{1}{N} \sum_{i=1}^N d_i - \bar{x}_i + \phi_i\left(\bar{x}_i\right)  = D_0 > 0.
		 \end{align*}
		 Therefore, it follows that $\lim_{t \rightarrow \infty} \bar{\xi}(t) = +\infty$. 
		 Thus,
		 \begin{align*}
			 \lim\limits_{t \rightarrow \infty} \lambda_i(t) = \lim\limits_{t \rightarrow \infty} \bar{\xi}(t) + Q_i\tilde{\xi}(t) = +\infty
		 \end{align*}
		 and hence $\lambda_i(t)$ diverges to $+\infty$.
		 
		 From \eqref{eq:optimal_power_extend}, it holds that $dg_i^m(\bar{\xi})/d\lambda$ is a constant for sufficiently large $\bar{\xi}$. Therefore, it can be verified from \eqref{eq:dyn_trans_1} and the boundedness of $\tilde{\xi}$ that $\dot{\bar{\xi}}$ satisfies
		 $$
		 \lim\limits_{t \rightarrow \infty}\dot{\bar{\xi}}(t) = D_0.
		 $$
		 In addition, boundedness of $\tilde{\xi}$ and divergence of $\bar{\xi}$ implies
		 \begin{align*}
			 \lim\limits_{t \rightarrow \infty} R^TG(1_N\bar{\xi} + Q\tilde{\xi}) & = R^T (d - \bar{x} + \phi(\bar{x})).
		 \end{align*}
		 In particular, $\tilde{\xi}(t)$ converges to a constant value. Therefore,
		 \begin{align*}
			 \lim\limits_{t \rightarrow \infty} \dot{\lambda}_i(t) = \dot{\bar{\xi}}(t) + Q\dot{\tilde{\xi}} = D_0
		 \end{align*}
		 which completes the proof.
	\end{proof}
	
	\section{Simulation} \label{sec:sim}
	\subsection{Robustness to Changes}	
	For the simulation, the continuous-time algorithm \eqref{eq:dist_dyn} is discretized using the forward difference method.
	In particular, if we denote the sampling period with $T$, then the distributed algorithm \eqref{eq:dist_dyn} can be discretized into 
	\begin{multline} \label{eq:forward}
		\lambda_i^\rd((q+1)T) =  \lambda_i^\rd(qT) + T \frac{dg_i^m}{d\lambda} (\lambda_i^\rd(qT)) \\
		+ T k \sum_{j \in \cN_i} (\lambda_j^\rd(qT) - \lambda_i^\rd(qT)),
	\end{multline}
	for $i \in \cN$, where $q\geq0$ is integer and $\lambda_i^\rd$ is the discretized state.
	In particular, $\lambda_i^\rd(t)$ is a piecewise constant, right continuous signal which is updated every $T$ seconds.
	Specifically, the value of $\lambda_i^\rd(t)$ is held constant until the next update.

	Numerical simulation is done with IEEE 30 bus system \cite{Alsac1974} to verify the proposed algorithm. 
	The local cost function is given by $f_i(x_i) = a_i + b_i x_i + c_i x_i^2$ where $b_i$ and $c_i$ are strictly positive for $i \in \{1,2,5,8,11,13\}$ which are the buses with a generator. 
	The power demand of each bus satisfies $d_i \in [0,94.2]$ and $\sum_{i=1}^N d_i = 283.4$. 
	The loss function is chosen as $\phi_i(x_i) = \alpha_i x_i^2$ where $\alpha_i \in [0.0001,0.0007]$ are chosen randomly. 
	It is assumed that each bus is capable of running the proposed algorithm \eqref{eq:forward} and that two buses connected by a branch can communicate. 
	Coupling gain of $k=40$ is used for the simulation. 
	\rt{For the implementation, \eqref{eq:forward} is used with the sampling time of $0.005$ seconds. }%
	We consider the following scenario:
	\begin{itemize}
		\item[S1)] Normal operation condition for $0 \leq t \leq 10s$.
		\item[S2)] At $t=10s$, demand at bus $5$ is decreased by $20\%$.
		\item[S3)] At $t=20s$, generator at bus $1$ stops generation and leaves the network.
		\item[S4)] At $t=30s$, bus $1$ joins the network again and the maximum generation at bus $8$ increases by $20\%$.
	\end{itemize}
	\rt{
	Note in particular that the above scenarios are simulated in one continuous session. 
	Moreover, network topology changes as a node leaves and joins the network during the operation.
	}
	
	Simulation results are shown in Fig. \ref{fig:graph1}. 
	It can be seen that an optimal solution is obtained in a distributed manner \rt{despite the changes in operation conditions}. 
	At $t=20s$, the problem  \eqref{eq:prim_prob} becomes infeasible due to the lack of generation at bus $1$. 
	Hence, the trajectory of $\lambda_i^\rd(t)$ diverges which verifies the result of Theorem \ref{thm:infeas}. 
	As feasibility is recovered at $t=30s$, $\lambda_i^\rd(t)$ converges again.
	Value of the cost function is also shown in Fig. \ref{fig:graph1}(c) and it is seen that the optimal cost is approximately recovered.
	The trajectory of power mismatch $\sum_{i=1}^N d_i - x_i(t) + \phi_i(x_i(t))$ is shown in Fig. \ref{fig:graph1}(d).	
	It is observed that the mismatch converges to zero (except for the case when the problem is infeasible) implying supply-demand balance is satisfied.
	
	From the repeated simulations, we also observed that the coupling gain $k$ and the sampling time $T$ have a close relationship for the stability. 
	Specifically, if the gain $k$ is large, the sampling time must be reduced to yield a stable algorithm.
	For example, it is observed that \eqref{eq:forward} diverges with $T = 0.01$ but converges if $T = 0.005$.
	\begin{figure}[h]
		\centering
		\begin{subfigure}[b]{\columnwidth}
			\centering
			\includegraphics[scale=0.19]{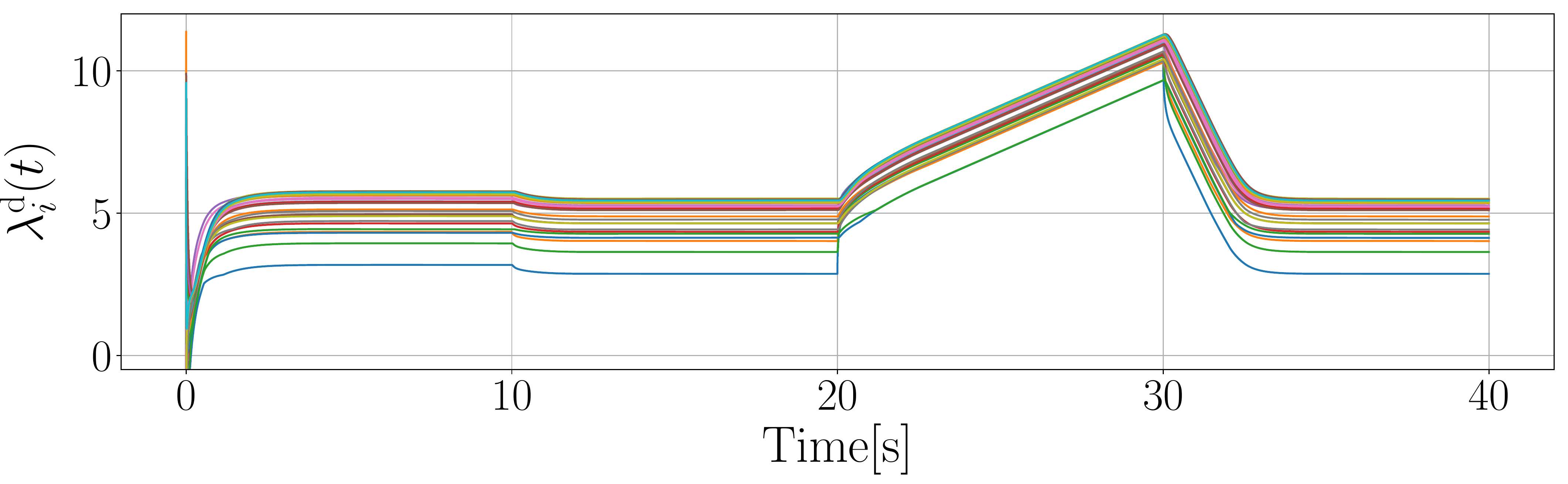}
			\caption{Trajectories of $\lambda_i^\rd(t)$}
			\vspace{1em}
		\end{subfigure}
		\begin{subfigure}[b]{\columnwidth}
			\centering
			\includegraphics[scale=0.19]{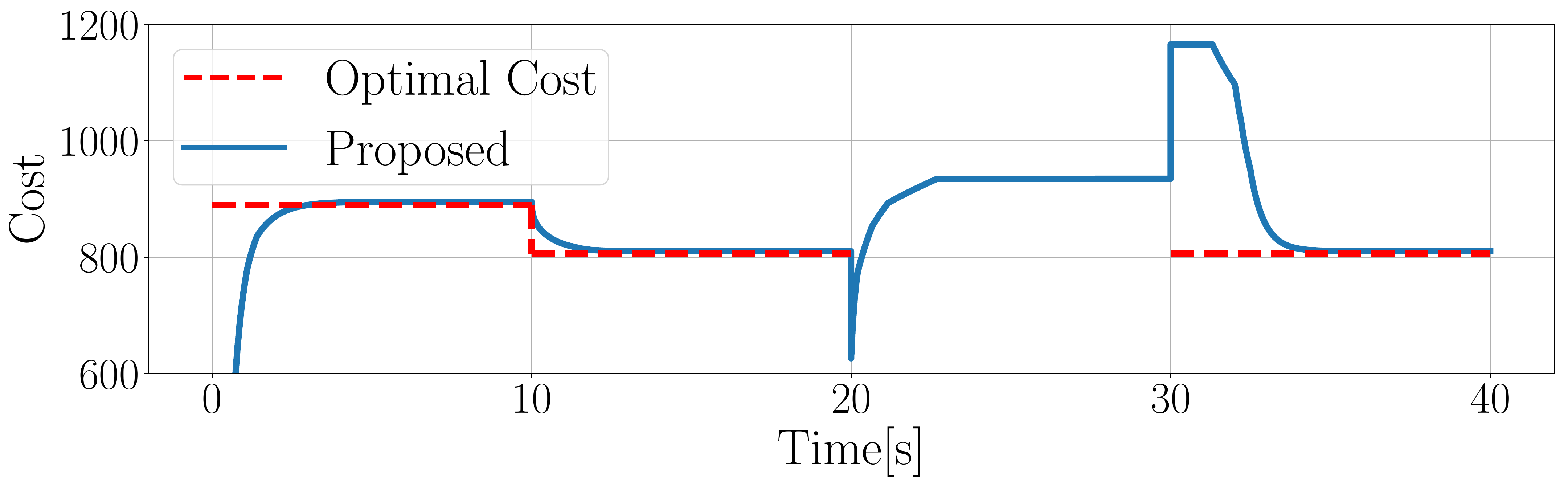}
			\caption{Trajectory of $\sum_{i=1}^N f_i(\hat{x}_i(t))$ and optimal cost}
			\vspace{1em}
		\end{subfigure}
		\begin{subfigure}[b]{\columnwidth}
			\centering
			\includegraphics[scale=0.19]{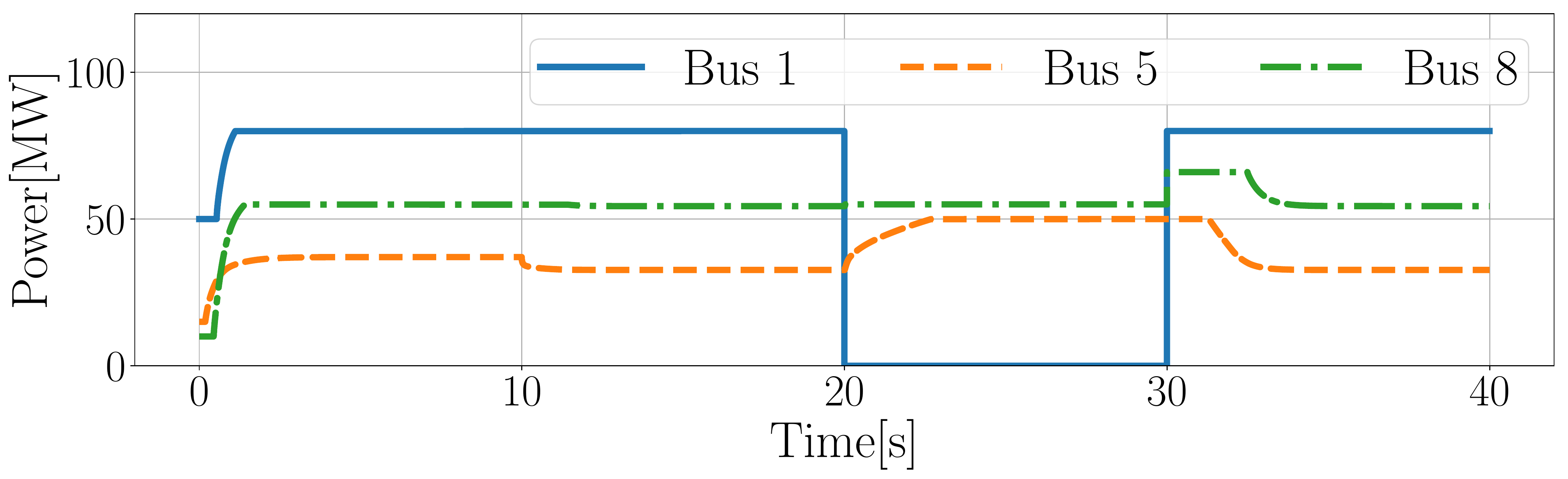}
			\caption{Trajectories of $\hat{x}_i(t)$ for selected buses}
			\vspace{1em}
		\end{subfigure}
		\begin{subfigure}[b]{\columnwidth}
			\centering
			\includegraphics[scale=0.19]{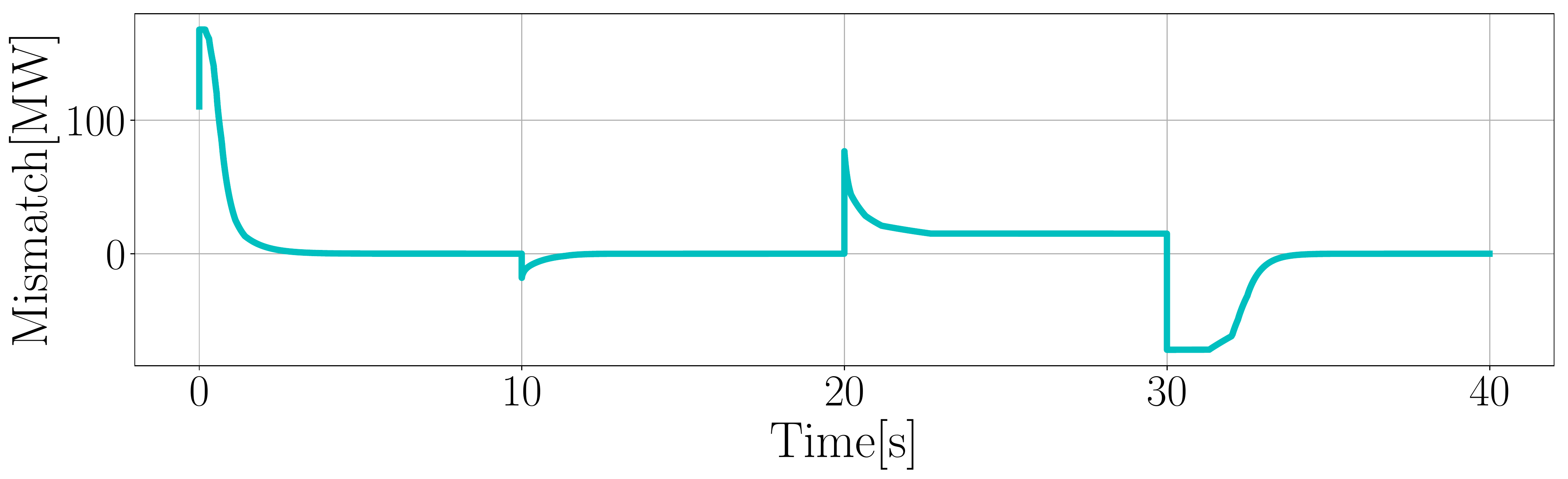}
			\caption{Trajectory of power mismatch $\sum_{i=1}^N d_i - x_i(t) + \phi_i(x_i(t))$}
		\end{subfigure}
		\caption{ Simulation results} \label{fig:graph1}
	\end{figure}
	
	\subsection{Asynchronous Update}	
	The discretized algorithm \eqref{eq:forward} supposes that all agents use the same sampling time $T$ and that the update is made synchronously across all agents. 
	However, it may be hard to implement the synchronous algorithms in practice due to the distributed nature of the system.
	Instead, in this section, we consider the case where each agent uses different sampling time $T_i$ and that the update is done asynchronously. 
	In particular, let $q \geq 0$ be an integer representing the time index. Then, the distributed algorithm \eqref{eq:dist_dyn} becomes
	\begin{align} \label{eq:forward_asynch}
		\lambda_i^\rd((q+1)T_i) & =  \lambda_i^\rd(qT_i) + T_i \frac{dg_i^m}{d\lambda} (\lambda_i^\rd(qT_i)) \nonumber\\
		& \qquad + T_i k \sum_{j \in \cN_i} \left(\lambda_j^\rd(qT_i) - \lambda_i^\rd(qT_i) \right)
	\end{align}
	where $\lambda_i^\rd(t)$ is a piecewise constant, right continuous signal which is updated every $T_i$ seconds.
	Note that the algorithm \eqref{eq:forward_asynch} is equivalent to \eqref{eq:forward} if $T_i = T$ for all $i \in \cN$.
	However, if sampling time is different between agents, then some agents update more frequently than others.
		
	The same IEEE 30 bus system is used to simulate the algorithm \eqref{eq:forward_asynch}. 
	For the simulation, gain of $k=20$ is used while sampling time of $0.009s$ is used for all agents except for the ones denoted in Table \ref{table:sampling_time}.
	Specifically, we consider the case when some nodes update less frequently.
	Simulation results are shown in Fig. \ref{fig:graph2}. 
	Trajectories of $\lambda_i^\rd(t)$ for selected bus is shown in Fig. \ref{fig:graph2}(a) and Fig. \ref{fig:graph2}(b).
	It is clearly seen that each variable is updated asynchronously, and that some agents are updated more frequently. 
	Nonetheless, the solution $\lambda_i^\rd(t)$ still converges.
	Additionally, Fig. \ref{fig:graph2}(c) depicts that the converged solution satisfies supply and demand balance.
	
	\begin{table}
		\def\arraystretch{1.5}%
		\centering
		\begin{tabular}{c|c||c|c}
			Bus & Sampling Time (s) & Bus & Sampling Time (s) \\ \hline
			Bus $5$ & $0.05$ & Bus $11$ & $0.03$ \\ 
			Bus $16$ & $0.07$ & Bus $17$ & $0.02$ \\ 
			Bus $21$ & $0.04$ \\ 
		\end{tabular}
		\caption{Sampling time used for each node.}
		\label{table:sampling_time}
	\end{table}
		
	\begin{figure}[h]
		\centering
		\begin{subfigure}[b]{\columnwidth}
			\centering
			\includegraphics[scale=0.19]{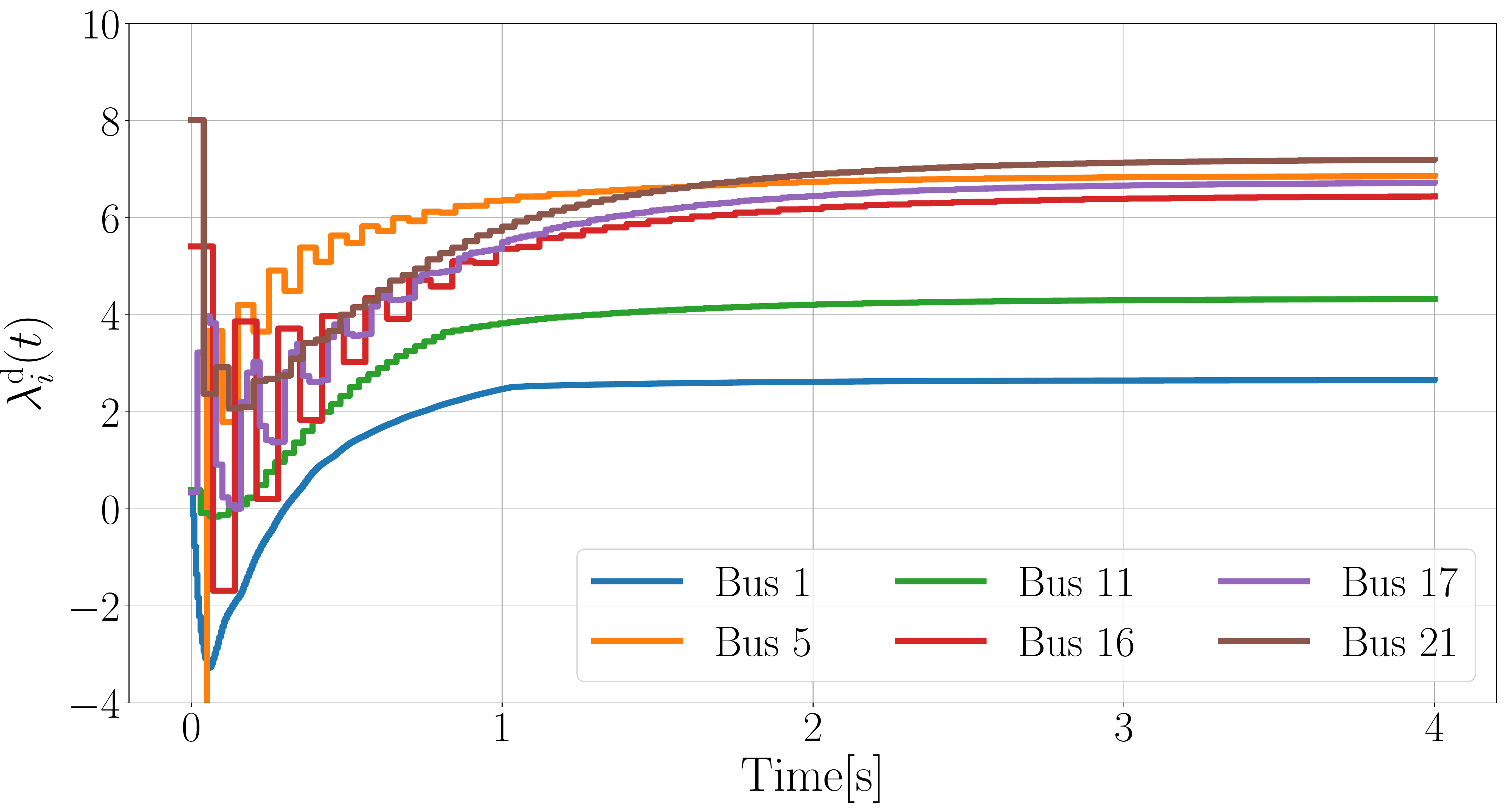}
			\caption{Trajectories of $\lambda_i^\rd(t)$}
			\vspace{1em}
		\end{subfigure}
		\begin{subfigure}[b]{\columnwidth}
			\centering
			\includegraphics[scale=0.19]{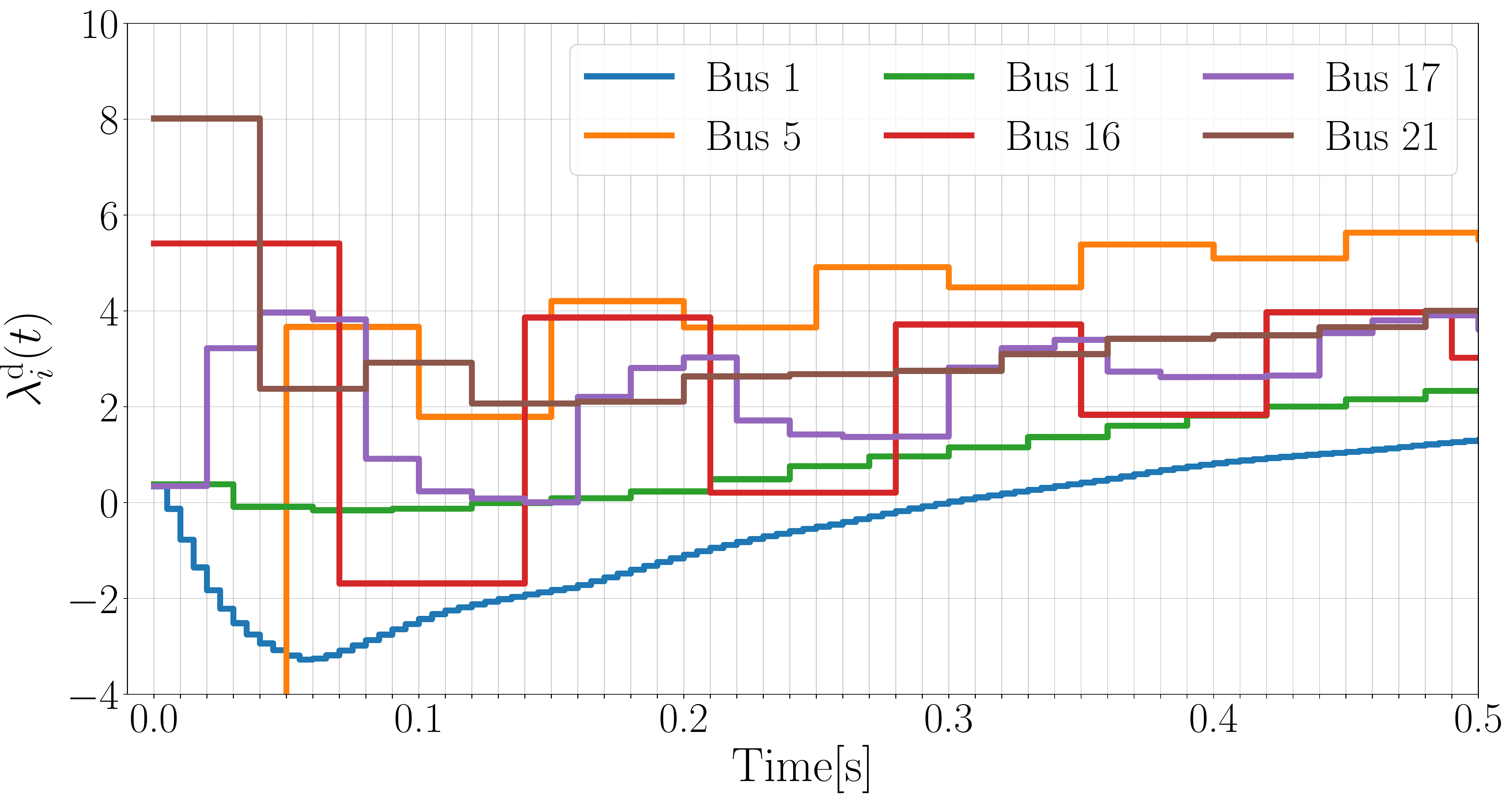}
			\caption{Trajectories of $\lambda_i^\rd(t)$ for $0 \leq t \leq 0.5$}
			\vspace{1em}
		\end{subfigure}
		\begin{subfigure}[b]{\columnwidth}
			\centering
			\includegraphics[scale=0.19]{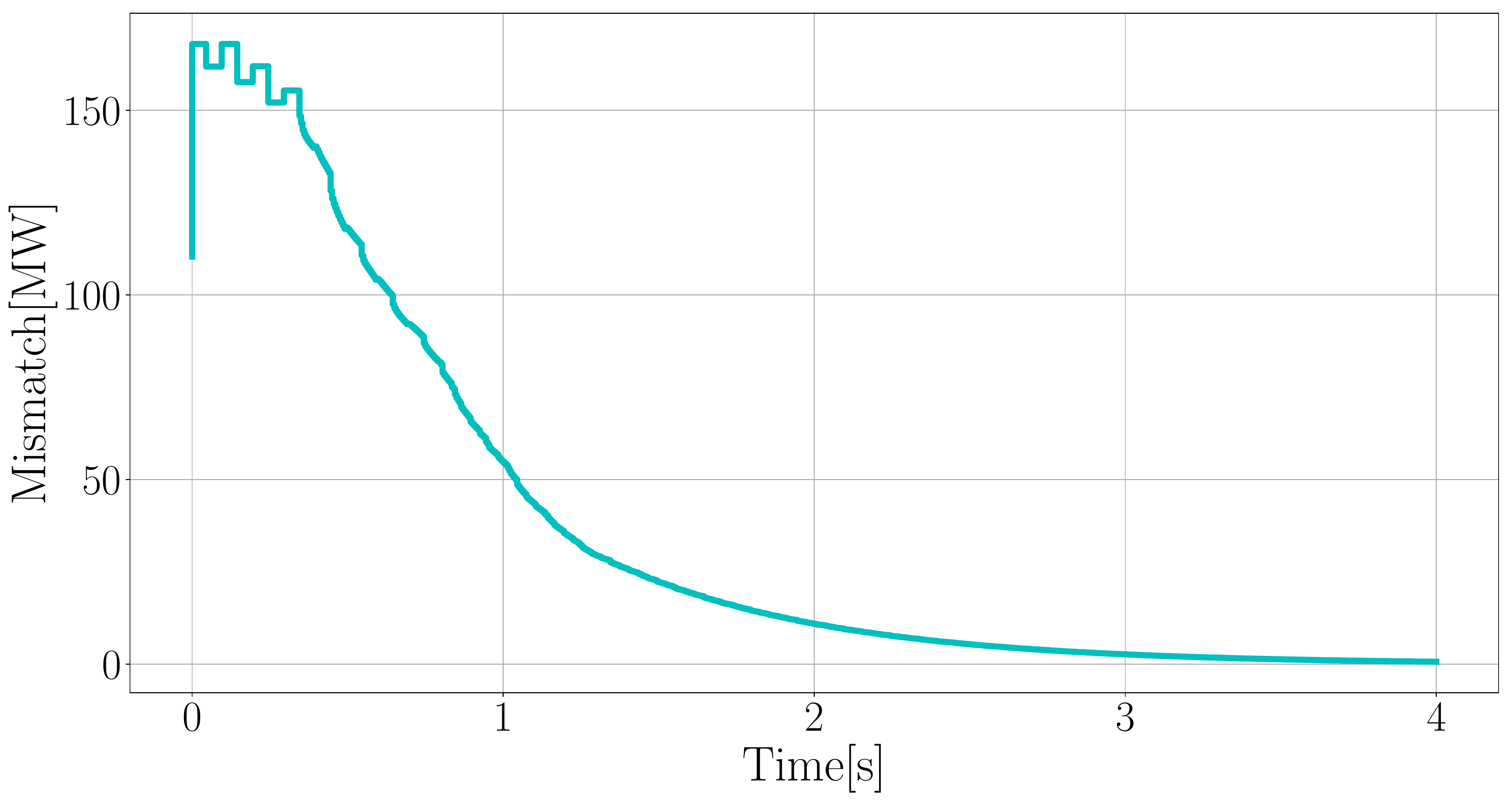}
			\caption{Trajectories of power mismatch $\sum_{i=1}^N d_i - x_i(t) + \phi_i(x_i(t))$}
			\vspace{1em}
		\end{subfigure}
		\caption{ Simulation results for asynchronous update.} \label{fig:graph2}
	\end{figure}
			
	\section{Conclusion and Future Works} \label{sec:con}
	The economic dispatch problem with nonlinear, separable power losses has been studied in this paper. 
	Due to the addition of nonlinear loss, the EDP becomes a non-convex optimization problem. 
	However, it has been shown that convex relaxation with dual decomposition can be used to obtain an optimal solution. 
	The distributed algorithm is proposed and it is shown to converge to a feasible solution while an optimal solution is recovered with sufficiently high coupling gain. 
	Specifically, the proposed algorithm does not require any initialization process and converges from any initial condition.
	Moreover, the behavior of the proposed algorithm is analyzed when the problem is infeasible. 
	\rt{
		Future works include the theoretical analysis of the discretized version of the proposed algorithm.
	}

\end{document}